\documentclass[sigconf]{acmart}
\usepackage{cancel}
\usepackage[ruled,vlined]{algorithm2e}
\usepackage{enumitem}
\usepackage{bm}
\usepackage{xcolor}
\usepackage{subfig}
\usepackage{graphicx}
\usepackage{listings}
\usepackage{xcolor}

\usepackage[normalem]{ulem}
\useunder{\uline}{\ul}{}

\AtBeginDocument{%
  }

\setcopyright{acmlicensed}
\copyrightyear{2025}
\acmYear{2025}
\acmConference[KDD '25]{Proceedings of the 31st ACM SIGKDD Conference on Knowledge Discovery and Data Mining V.2}{August 3--7, 2025}{Toronto, ON, Canada}
\acmDOI{XXXXXXX.XXXXXXX}
\acmISBN{978-1-4503-XXXX-X/2018/06}

\newcommand{\ie}{\textit{i.e.}}
\newcommand{\eg}{\textit{e.g.}}

\newcommand{\adj}{\tilde{A}}
\newcommand{\adjplus}{\tilde{A}_+}
\newcommand{\users}{\mathcal{U}}
\newcommand{\items}{\mathcal{I}}

\newtheorem{assumption}{Assumption}[section]

\begin{document}

\title[Collaborative Filtering Meets Spectrum Shift]{
Collaborative Filtering Meets Spectrum Shift: Connecting User-Item Interaction with Graph-Structured Side Information
}

\author{Yunhang He}
\authornote{Both authors contributed equally.}
\email{yhhe2004@gmail.com}
\author{Cong Xu}
\authornotemark[1]
\email{congxueric@gmail.com}
\affiliation{%
  \institution{East China Normal University}
  \city{Shanghai}
  \country{China}
}

\author{Jun Wang}
\affiliation{%
  \institution{East China Normal University}
  \city{Shanghai}
  \country{China}
}
\email{wongjun@gmail.com}

\author{Wei Zhang}
\authornote{Corresponding author.}
\affiliation{%
  \institution{East China Normal University \\ \& Shanghai Innovation Institute}
  \city{Shanghai}
  \country{China}
}
\email{zhangwei.thu2011@gmail.com}



\begin{abstract}
  Graph Neural Networks (GNNs) have demonstrated their superiority in collaborative filtering,
  where the user-item (U-I) interaction bipartite graph serves as the fundamental data format.
  However, when graph-structured side information (\eg, multimodal similarity graphs or social networks) 
  is integrated into the U-I bipartite graph,
  existing graph collaborative filtering methods fall short of achieving satisfactory performance.
  We quantitatively analyze this problem from a spectral perspective.
  Recall that a bipartite graph possesses a full spectrum within the range of [-1, 1], with the highest frequency exactly achievable at $-1$ and the lowest frequency at $1$;
  however, we observe as more side information is incorporated, 
  the highest frequency of the augmented adjacency matrix progressively shifts rightward.
  This spectrum shift phenomenon has caused previous approaches built for the full spectrum [-1, 1] 
  to assign mismatched importance to different frequencies.
  To this end,
  we propose Spectrum Shift Correction (dubbed SSC), incorporating \textit{shifting} and \textit{scaling} factors 
  to enable spectral GNNs to adapt to the shifted spectrum.
  Unlike previous paradigms of leveraging side information, 
  which necessitate tailored designs for diverse data types, 
  SSC directly connects traditional graph collaborative filtering with any graph-structured side information.
  Experiments on social and multimodal recommendation demonstrate the effectiveness of SSC, 
  achieving relative improvements of up to 23\% without incurring any additional computational overhead.
  Our code is available at https://github.com/yhhe2004/SSC-KDD.
\end{abstract}




\received{February 2025}
\received[accepted]{May 2025}

\maketitle

\section{Introduction}

The vast amount of information in the Internet age has imposed a significant cognitive load on individuals.
It is essential for online service providers~\cite{chen2019behavior,el2022twhin}, 
such as short video platforms and e-commerce websites, 
to curate and deliver information streams that align closely with user preferences.
Hence, recommender systems emerge as the core technology for this purpose.
The involved recommendation engine is continuously evolving so as to fully utilize the diverse types of perpetually accumulating data.
Among these, interaction data, 
which records whether a user has previously interacted with an item, 
serves as the most commonly studied data format~\cite{rendle2012bpr,he2020lightgcn}.
Apart from this fundamental user-item data,
there also exist abundant user-user relationships~\cite{ma2008sorec,wu2020diffnet++} such as friends and colleagues,
as well as diverse item multimodal data~\cite{he2015vbpr,zhang2021lattice,zhou2023freedom} including textual descriptions and visual thumbnails.
This side information serves as valuable complements to interaction data (especially for cold-start users and items~\cite{li2019both})
when group preferences and item similarities are accurately extracted.
Then a natural research direction arises, \ie, how to utilize interaction and side information together to match with real-world recommendation scenarios.

Due to the bipartite graph nature of the user-item interaction data,
graph neural networks (GNNs)~\cite{kipf2016gcn,gilmer2017neural,wang2019ngcf}, particularly spectral GNNs~\cite{defferrard2016chebynet,gasteiger2018appnp}, 
prove the superiority of collaborative filtering.  
For instance, LightGCN~\cite{he2020lightgcn} and JGCF~\cite{guo2023jgcf} achieve competitive recommendation performance despite their simplified architectures.
However, the models in social or multimodal recommendation tend to exhibit significant complexity when integrating aforementioned side information. 
Note that in social recommendation 
the user-user social network is the most crucial side information,
while in multimodal collaborative filtering,
the additional data is typically preprocessed into item-item similarity graphs~\cite{zhang2021lattice,zhou2023freedom} based on encoded multimodal features.
To fully exploit the derived graph-structured side information, 
researchers typically develop dedicated modules~\cite{yu2023mgcn,jiang2024share} or auxiliary loss functions~\cite{wei2023mmssl,Zhou2023bm3,wang2023dsl}.
These interventions however limit the generalizability of the model to tasks 
with varying \textit{characteristics}, \textit{statistical properties}, and \textit{signal-to-noise ratios}.
In view of the effectiveness of traditional GNNs in collaborative filtering, 
we are to investigate \textbf{how to extend these spectral GNNs to integrate graph-structured side information}.

A technically feasible solution is to apply these GNNs to an augmented graph that
integrates the user-user social networks or the item-item similarity graphs into the original user-item bipartite graph.
Although certain performance improvements can be achieved due to the mitigation of the graph sparsity issue,
they remain inferior to those specialized models with high complexity~\cite{wang2023dsl,wei2023mmssl}.
Recall that spectral GNNs are inherently a series of graph filters 
designed to assign appropriate weights to the spectrum of the adjacency matrix under consideration.
For the normalized adjacency matrix derived from a user-item bipartite graph,
its spectrum spans the interval [-1, 1], with the highest frequency exactly achievable at -1
and the lowest frequency at 1.
Consequently, previous work is largely based on the partially reasonable \textit{full-spectrum assumption} 
(as illustrated in \figurename~\ref{fig-filters}),
which becomes problematic when considering an augmented graph:
As more graph-structured side information is integrated,
the highest frequency of the augmented adjacency matrix progressively shifts rightward.
This spectrum shift phenomenon is validated both empirically and theoretically, 
regardless of whether a $\kappa$-rescaling or $\kappa$-nearest-neighbors mechanism 
is employed for graph integration (see Section~\ref{section-spectrum-shift}).
Due to the mismatched spectrum, 
previous approaches designed for the full spectrum are unable to accurately assign the intended importance,
leading to suboptimal recommendation performance.

\begin{figure}[t]
  \centering
  \includegraphics[width=0.47\textwidth]{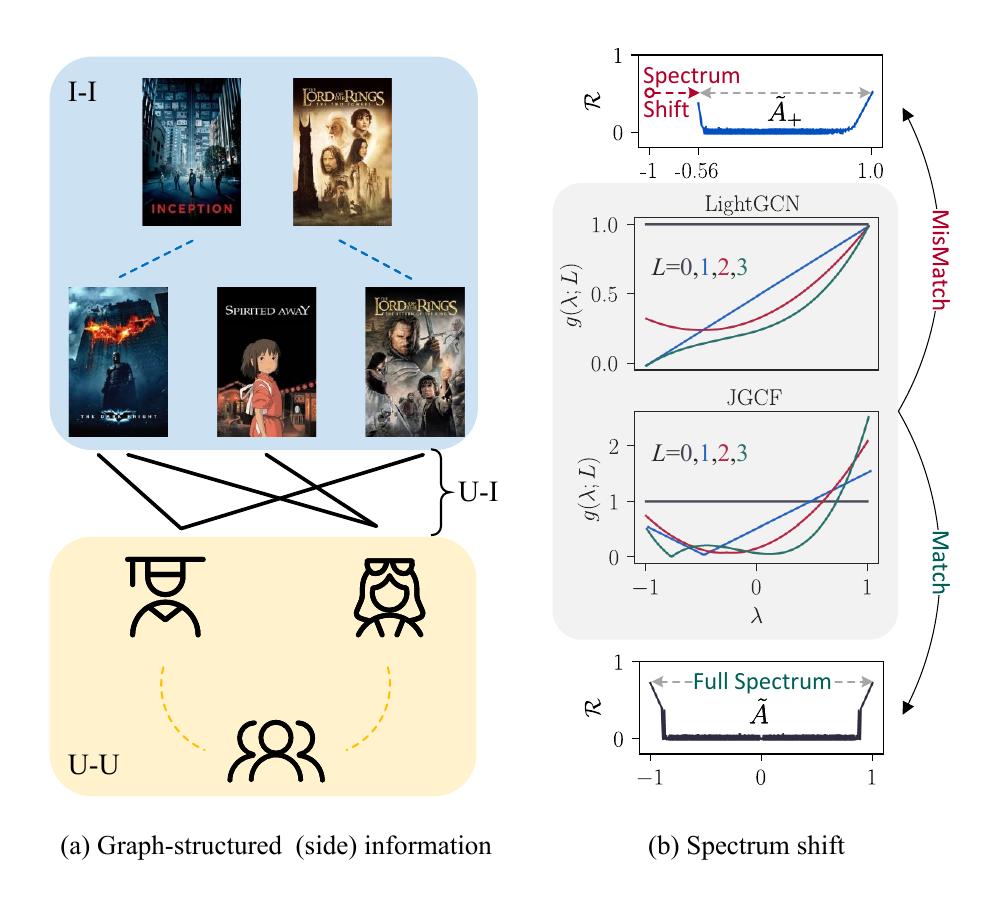}
  \caption{
    The integration of (a) graph-structured side information results in (b) spectrum shift.
    LightGCN and JGCF are tailored for the full spectrum of $\adj$, 
    which inevitably results in a mismatch when applied to the augmented adjacency matrix $\adjplus$.
  }
  \label{fig-filters}
\end{figure}

In this paper, we present Spectrum Shift Correction (SSC) to correct the mismatched spectrum. 
This approach effectively transforms the right-shifted spectrum back to the `predefined' domain
by setting appropriate shifting and scaling factors.
It exhibits the following advantages:
\textbf{1)} 
SSC over the spectrum is mathematically equivalent to applying the same operations directly to the augmented adjacency matrix.  
Therefore, SSC introduces no additional computational overhead.
\textbf{2)} 
SSC can be seamlessly integrated into existing spectral GNNs in a plug-and-play manner. 
Both LightGCN, which emphasizes low frequencies, 
and JGCF, driven by Jacobi Polynomial Bases~\cite{wang2022jacobiconv}, gain significant improvements from the application of SSC.
\textbf{3)} 
Side information often contains noise, as it typically includes more recommendation-unrelated information compared to interaction data. 
Fortunately, SSC demonstrates strong robustness to such noise, provided that the shifting and scaling factors are carefully tuned.
It is worth noting that although SSC can be covered by learnable spectral GNNs in theory, 
they struggle with training difficulties and instability in recommendation scenarios.
As evidenced by the results presented in \tablename~\ref{table-performance-multimodal} and \ref{table-performance-social}, 
LightGCN and JGCF equipped with SSC are able to achieve state-of-the-art performance across five public datasets, 
encompassing three multimodal collaborative filtering tasks and two social recommendation scenarios.

Our main contribution can be summarized as follows.
\begin{itemize}[leftmargin=*]
  \item 
  To the best of our knowledge, 
  we are the first to identify the spectrum shift phenomenon, 
  which hampers full utilization of graph-structured side information with existing spectral GNNs.
  \item 
  We propose Spectrum Shift Correction to enable the spectral GNNs to adapt to the shifted spectrum.
  This approach is model-agnostic and simple to implement with only a few lines of additional code.
  \item SSC enhances existing LightGCN and JGCF with no additional computational overhead.
  The effectiveness is validated in social and multimodal recommendation scenarios,
  achieving significant relative improvements of up to 23\%.
  Moreover, the augmented graph combined with SSC demonstrates robustness against noise within the side information.
\end{itemize}

\section{Preliminaries}

\subsection{Collaborative Filtering}

Collaborative filtering~\cite{rendle2012bpr,he2020lightgcn} is well-regarded for its remarkable recommendation performance, 
mostly relying on the sparse interaction matrix $R \in \mathbb{R}^{|\users| \times |\items|}$ between users $\users$ and items $\items$.
In this paper, 
we focus on implicit feedback; 
that is, $R_{ui} = 1$ if the user $u$ has clicked/purchased item $i$ previously, and $R_{ui} = 0$ otherwise.
A line of recommendation is to treat the interaction data as a bipartite graph, 
with the adjacency matrix defined as follows
\begin{equation*}
   A = \left 
   [ \begin{matrix}
  0 & R \\
  R^T & 0 \\
  \end{matrix} \right ] \in \{0, 1\}^{(|\users| + |\items|) \times (|\users| + |\items|)}.
\end{equation*}
Naturally,
Graph Neural Networks (GNNs) have become a prevalent technique in collaborative filtering 
due to their expertise in mining graph-structured data. 
For example, LightGCN~\cite{he2020lightgcn} is a simple yet effective graph collaborative filtering model:
\begin{equation*}
    H = \sum_{l=0}^{L} \frac{1}{L+1} \adj^l E \in \mathbb{R}^{(|\users| + |\items|) \times d},
\end{equation*} 
where $\adj = D^{-1/2} A D^{-1/2}$ denotes matrix after symmetric sqrt normalization,
and $E \in \mathbb{R}^{(|\users| + |\items|) \times d}$ is the trainable embeddings. 
Based on the final embeddings $H$, a recommendation loss function, such as the Bayesian Personalized Ranking (BPR) loss~\cite{rendle2012bpr},
is adopted to train the model.

\subsection{Spectral Graph Neural Networks}

Let $\tilde{A} = U\Lambda U^T$ represent the eigendecomposition of $\tilde{A}$,
wherein $U$ is the eigenvector matrix, 
and $\Lambda$ denotes the corresponding \textit{diagonal} matrix 
containing eigenvalues $\lambda_1 \le \cdots \le \lambda_{|\users| + |\items|}$, arranged from high-frequency to low-frequency~\cite{shen2021smoothing,guo2023jgcf}.
Intuitively, an eigenvector $U_{:, \lambda}$ corresponding to a low-frequency (\ie, larger) eigenvalue tends to assign similar values to connected nodes,
while eigenvectors associated with high-frequency (\ie, smaller) eigenvalues tend to exhibit greater variability across connected nodes.
Hence, graph signals would exhibit stronger low-frequency characteristics 
if they can predominantly be explained by eigenvectors associated with larger eigenvalues.

Spectral GNN~\cite{wang2022jacobiconv,guo2023favard} focuses on processing the graph signals in the spectral domain.
To be more specific, a spectral filter $g(\cdot)$ is \textit{element-wisely} applied to the diagonal entries of $\Lambda$:
\begin{equation}
    H = U g(\Lambda) U^T E = U \text{diag}\big(
        g(\lambda_1), \ldots, g(\lambda_{|\users| + |\items|})
    \big) U^T E.
\end{equation}
Following this paradigm, the adjustment of frequency characteristics can be highly effective.
For instance, 
the high-frequency signals can be completely eliminated through a low-pass filter $g(\lambda) = \lambda \cdot \mathbb{I}[\lambda \ge 0]$.
However, this explicit eigendecomposition of $\adj$ is computationally prohibitive, 
particularly in recommendation scenarios involving a substantial number of nodes.
Polynomial filters have thus gained popularity in spectral GNNs, 
as they can achieve a similar effect in an efficient spatial convolution manner.
Here, we illustrate the monomial of degree $L$:
\begin{equation}
    \label{eq-spectral-gnn-monomial}
    H = U g(\Lambda) U^T E = g(\adj) E = \sum_{l=0}^L \alpha_l \adj^l E, \: \text{if } g(\lambda; L) := \sum_{l=0}^L \alpha_l \lambda^l.
\end{equation}

Since $\lambda_{\min} = -1, \lambda_{\max} = 1$ holds true for any bipartite graph (see Appendix~\ref{appendix-proof-fsa} for a proof),
spectral GNNs applied in collaborative filtering are typically founded on the following full-spectrum assumption:
\begin{assumption}[Full-Spectrum Assumption]
  The spectrum of the normalized adjacency matrix spans the entire interval $[-1, 1]$.
\end{assumption}
Low- and high-frequency signals are widely recognized for their effectiveness in the majority of recommendation tasks~\cite{shen2021smoothing,peng2022less,zhu2024exploring}.
Consequently, prior successful graph collaborative filtering methods often place greater emphasis on frequencies close to -1 or 1.
It is readily observed that LightGCN is a special case of Eq.~\eqref{eq-spectral-gnn-monomial} with $\alpha_l = 1 / (L+1)$, 
promoting low-frequency signals as illustrated in \figurename~\ref{fig-filters}. 
In contrast, JGCF places additional emphasis on high-frequency signals.
Unfortunately, the full-spectrum assumption breaks down in the context of the augmented graph introduced below.
This greatly challenges the applicability of both LightGCN and JGCF.
Although recent advancements also suggest a learnable polynomial (\eg, Jacobi Polynomial Bases~\cite{wang2022jacobiconv}),
training difficulties and instability~\cite{he2020lightgcn,xu2023stablegcn} arise when applying this approach to collaborative filtering.
Hence, we are to endow existing spectral GNNs with a correction capability.

\subsection{Graph-Structured Side Information} \label{sec_graph_structured_side_information}

The real-world recommendation involves perpetually accumulating data of diverse types.
In collaborative filtering, graph-structured side information~\cite{zhang2021lattice,jiang2024share} has gained significant attention within the research community. 
Such information acts as a valuable supplement to interaction data, 
with the potential to alleviate the severe sparsity issue of bipartite graphs, 
particularly for cold-start users and items.
We observe two types of graph-structured side information frequently 
discussed in social recommendation and multimodal collaborative filtering (refer to Appendix~\ref{appendix-preprocessing} for preprocessing details), respectively.
\begin{itemize}[leftmargin=*]
    \item \textbf{U-U (user-user).} Modern media, especially social applications, encompass various types of relationships between users, such as friends or colleagues.
    The resulting social network is of particular value for recommendation 
    since the friends/colleagues may have overlapping life trajectories and similar commercial behaviors.
    Let $S_U \in \mathbb{R}^{|\users| \times |\users|}$ be the social network 
    and $S_U(\kappa) = \kappa S_U, \kappa \ge 0$ be the $\kappa$-rescaled variant used in the following.
    \item \textbf{I-I (item-item).}
    Compared to users, items generally have richer multimodal features, including textual descriptions and visual thumbnails.
    In the context of multimodal collaborative filtering, 
    it is observed that the state-of-the-art models typically integrate $\kappa$NN graphs constructed from respective modal features, 
    rather than direct utilization of the encoded features~\cite{zhang2021lattice,zhou2023freedom}. 
    For notation simplicity,
    let $S_I(\kappa), \kappa=0,1,2,\ldots$ represent the average graph over the $\kappa$NN modal similarity graphs.
\end{itemize}

Note that we introduce two commonly employed graph construction approaches here, 
namely $\kappa$-rescaling for $S_U(\kappa)$ and $\kappa$-nearest-neighbors for $S_I(\kappa)$. 
The explicit inclusion of the hyperparameter $\kappa$ enables us to quantitatively analyze the effectiveness of the graph-structured side information.
Since the side information may not always be beneficial, 
an ideal utilization method should exhibit robustness to the noise present within it 
and achieve a proper balance between the interaction data and the graph-structured side information.

\section{Spectrum Analysis on Augmented Graphs} \label{sec_augmented_graph}

To fully exploit the aforementioned graph-structured side information, 
researchers typically develop dedicated modules~\cite{zhou2023freedom,jiang2024share} or auxiliary loss functions~\cite{wei2023mmssl,wang2023dsl}.
These interventions however limit the generalizability.
For instance, 
when the multimodal similarity graph includes a larger number of $\kappa$-nearest neighbors, 
existing multimodal collaborative filtering models are significantly affected by the noise therein,
even performing worse than a simplified approach that does not utilize side information at all.

In view of the effectiveness of spectral GNNs in traditional collaborative filtering, 
we are to investigate how to extend these spectral GNNs to effectively integrate graph-structured side information.
If social networks or multimodal similarity graphs, or both, are available,
a straightforward approach is to integrate the graph-structured side information,
leading to an augmented graph:
\begin{equation}
    \label{eq-AAM}
 A_+ = \left [ \begin{matrix}
   S_{U}(\kappa) & R \\
   R^T & 0 \\
   \end{matrix} \right ]    
   \text{ or }
    \left [ \begin{matrix}
   0 & R \\
   R^T & S_{I}(\kappa) \\
   \end{matrix} \right ]    
   \text{ or }
    \left [ \begin{matrix}
   S_U(\kappa) & R \\
   R^T & S_{I}(\kappa) \\
   \end{matrix} \right ]
\end{equation}
Notably, this augmented graph holds significant value for two primary purposes:
1) It can seamlessly integrate both user-user and item-item graphs;
2) As empirically demonstrated in Section~\ref{section-exp-noise}, 
previous social or multimodal models exhibit sensitivity to noise in the side information, 
whereas the augmented graph combined with the proposed SSC demonstrates robustness against such noise.

Before delving into the details of our proposed method, 
it is necessary to first justify why previous spectral GNNs have seldom demonstrated superior performance compared to other specialized models. 
In this section, we conduct a quantitative analysis from a spectral perspective to shed light the underlying reasons. 
This analysis further inspires a simple yet effective modification to improve existing graph-based collaborative filtering methods.

\subsection{Oracle Spectrum Importance}

The design of the spectral filter should adhere to certain principles 
in order to ensure robust performance regardless of the augmented graph.
Remark that the goal of collaborative filtering is to predict the non-zero entries in the test graph $B \in \{0, 1\}^{(|\users| + |\items|) \times (|\users| + |\items|)}$,
which are the \textit{oracle} interactions (often occurred in the future) observed in the test set.
To determine the importance of frequency $\lambda$ and its corresponding eigenvector $U_{:, \lambda}$,
the absolute Rayleigh quotient\footnote{
  The Rayleigh quotient is inspired by the metric adopted in~\cite{guo2023jgcf}, while offering more intuitive insights.
}
is employed for a quantitative measurement: 
\begin{equation}
    \mathcal{R}(x; B) := \frac{|x^T B x|}{x^T x} = \frac{|\text{vec}(B) \cdot \text{vec}(xx^T)|}{x^T x} , \quad \forall \: x \in \mathbb{R}^{|\users| + |\items|},
\end{equation}
where $\text{vec}(\cdot)$ represents the vectorization operation.
Intuitively, a large $R(x; B)$ is obtained iff $x x^T$ closely resembles $B$ without regard to the sign.
In other words,
if the eigenvector $U_{:, \lambda}$ demonstrate a large Rayleigh quotient,
the final embeddings $H$ shall incorporate more $U_{:, \lambda}$ to facilitate the reconstruction of the test graph $B$.
Hence, $\mathcal{R}(U_{:, k}; B) \:, k=1,2,\ldots, |\users| + |\items|$, can be employed as an indication of the respective importance.

\subsection{Spectrum Shift Phenomenon}
\label{section-spectrum-shift}
In this part, we analyze from a spectral perspective 
why previous graph collaborative filtering models fail to deliver satisfactory results on the augmented graph.

\begin{figure}[t]
  \centering
  \includegraphics[width=0.47\textwidth]{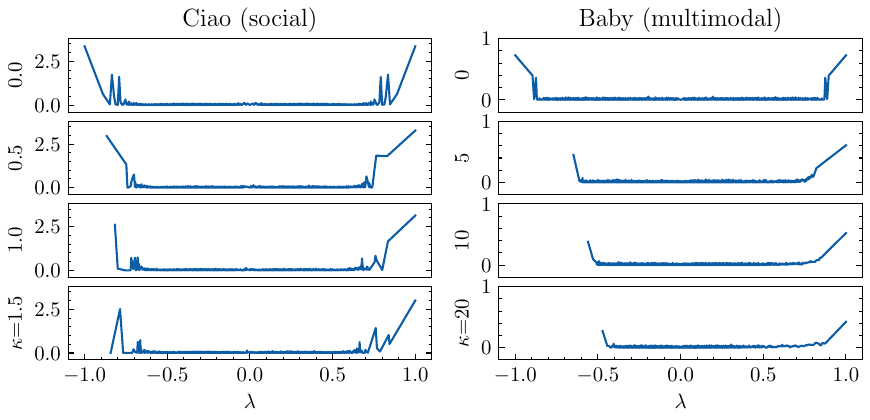}
  \caption{
    Oracle spectrum importance $\mathcal{R}(U_{:, \lambda}; B)$ as more side information is integrated (\ie, $\kappa$ increases).
  }
  \label{fig-oracle-spectrum}
\end{figure}

\textbf{Empirical observations on real-world datasets.}
We investigate how the absolute Rayleigh quotient $\mathcal{R}(U_{:,\lambda}; B)$ evolves 
as additional side information is integrated into the augmented adjacency matrix $\adjplus$. 
This exploration may reveal a universal framework for generalized graph collaborative filtering. 
Specifically, we consider an augmented adjacency matrix $\adjplus$ that integrates either $S_U(\kappa)$ or $S_I(\kappa)$, 
as they represent the two most prevalent approaches: 
$\kappa$-rescaling (for U-U) and $\kappa$-nearest-neighbors (for I-I). 
We have the following two key findings from \figurename~\ref{fig-oracle-spectrum}:
\begin{itemize}[leftmargin=*]
  \item 
  The full spectrum, spanning the interval [-1, 1], 
  undergoes compression with a fixed point at $\lambda_{\max} \equiv 1$. 
  As $\kappa$ increases, the minimum eigenvalue $\lambda_{\min}$ shifts rightward 
  regardless of $S_U(\kappa)$ or $S_I(\kappa)$ integrated.
  Therefore, the full-spectrum assumption no longer holds true for an augmented graph.
  \item
  In comparison to $S_U$, a more significant rightward shift is observed when $S_I$ is incorporated.
  The $\kappa$-nearest-neighbors mechanism appears to result in a more pronounced mismatch for spectral GNNs,
  consequently worse performance.
\end{itemize}

\textbf{Theoretical analysis.}
We theoretically justify the spectrum shift phenomenon 
and elucidate the distinctions between $\kappa$-rescaling and $\kappa$-nearest neighbors.
\begin{theorem}
  \label{thm-spectrum-shift}
  Let $S(\kappa) \in \mathbb{R}^{|\mathcal{U}| \times |\mathcal{U}|}$ denote a symmetric matrix with non-zero row sums.
  For the augmented adjacency matrix $\adjplus$,
  its maximum eigenvalue value $\lambda_{\max}(\kappa) \equiv 1$.
  Moreover, when $\kappa \rightarrow +\infty$
  \begin{itemize}[leftmargin=*]
    \item If $S(\kappa) = \kappa S$ is defined in a $\kappa$-rescaling manner, it follows that
    \begin{equation}
      \lambda_{\min}(\adjplus(\kappa)) = \min(\lambda_{\min} (\tilde{S}), 0).
    \end{equation}
    \item If $S(\kappa)$ is defined in a $\kappa$-nearest-neighbors manner, it follows that
    \begin{equation}
      \lambda_{\min}(\adjplus(\kappa)) \approx  0.
    \end{equation}
  \end{itemize}
  The same results can be drawn for I-I graphs.
\end{theorem}
\begin{proof}
  The results follow from the observation that $\adjplus(\kappa)$ converges to an adjacency matrix determined by $\tilde{S}$ (w.r.t. $\kappa$-rescaling) or a positive semidefinite matrix (w.r.t. $\kappa$-nearest-neighbors) as $\kappa \rightarrow +\infty$. Additional details can be found in Appendix~\ref{appendix-proof-ss}.
\end{proof}
Note that both social networks and multimodal $\kappa$NN graphs are generally more homogeneous than bipartite graphs, 
thereby resulting in a minimum eigenvalue strictly greater than -1. 
Furthermore, 
there exist significant differences of the application of $S_U$ and $S_I$,
primarily due to their distinct construction methods rather than differences in user/item side information.

\textbf{Foreseeable implication.}
Because of the mismatch between the shifted spectrum and the full-spectrum assumption,
which is fundamental to existing spectral GNNs in collaborative filtering, 
these models fail to fully leverage the side information, 
ultimately leading to performance inferior to that of specialized models.
Of course, one could repeatedly conduct spectrum analysis and redesign appropriate graph filters, as demonstrated by LightGCN and JGCF over the full spectrum. 
This protocol is computational prohibitive and far from practice.

\section{Spectrum Shift Correction}

\begin{algorithm}[t]
  \caption{
    Python-style algorithm for factor estimation
  }
  \label{alg-factor}
  \begin{lstlisting}[basicstyle=\small\ttfamily]
max_eigen_val = 1.
min_eigen_val = 1 - power_iteration(I - A, 
                num_iterations=100)
mu = (min_eigen_val + max_eigen_val) / 2
delta = (max_eigen_val - min_eigen_val) / 2
\end{lstlisting}
\end{algorithm}

\subsection{Simple yet Effective Method}

\begin{table}[h]
  \centering
  \caption{
    Learnable graph filters versus LightGCN.
  }
  \label{table-learnable}
\begin{tabular}{lcccc}
  \toprule
\multicolumn{1}{c}{} & \multicolumn{2}{c}{Ciao}                                              & \multicolumn{2}{c}{Baby}                                              \\
\multicolumn{1}{c}{} & R@20 & N@20 & R@20 & N@20 \\
\midrule
ChebyNet             & 0.0722                            & 0.0375                            & 0.0609                            & 0.0259                            \\
JacobiConv           & 0.0726                            & 0.0379                            & 0.0716                            & 0.0316                            \\
\midrule
LightGCN             & 0.0770                            & 0.0381                            & 0.0851                            & 0.0372                           \\
\quad \quad +SSC             & 0.0861                            & 0.0431                            & 0.1044                            & 0.0458                           \\
\bottomrule
\end{tabular}
\end{table}

While the ideal solution for the aforementioned challenge 
is to learn optimal (polynomial) graph filters~\cite{chien2021gprgnn,wang2022jacobiconv,guo2023jgcf,guo2023favard} automatically, 
training difficulties~\cite{he2020lightgcn} are encountered in collaborative filtering, 
leading to an unstable training process~\cite{xu2023stablegcn} 
and sometimes inferior performance (see \tablename~\ref{table-learnable}).
Here we introduce a simple but effective method of Spectrum Shift Correction (dubbed SSC) to endow existing unlearnable graph filters with a correction capability.
Specifically, SSC transforms shifted spectrum back to the pre-defined domain of $g$ as follows
\begin{equation}
  \label{eq-SSC-eigenvalue}
  g_+(\lambda) = g \circ \phi(\lambda), \quad \phi(\lambda) := \frac{\lambda - \mu}{\Delta},
\end{equation}
where $\mu \in [0, 1]$ and $\Delta \in (0, 1]$ are the \textit{shifting} and \textit{scaling} factors, respectively.
As can be seen,
the full spectrum [-1, 1] could be `recovered' if $\mu = (\lambda_{\max} + \lambda_{\min}) / 2$ and $\Delta = (\lambda_{\max} - \lambda_{\min}) / 2$.
For example, if the current spectrum spans the interval [-0.6, 1], 
setting $\mu = 0.2$ and $\Delta = 0.8$ exactly restores the adjusted spectrum to the range [-1, 1].
In this case, the modified graph filters can be expected to associate different frequencies with intended importance.

It is worth mentioning that the explicit eigendecomposition is not required.
Eq.~\eqref{eq-SSC-eigenvalue} is equivalent to the SSC operation directly applying to the augmented adjacency matrix $\adjplus$; that is,
\begin{equation} 
  \label{eq-SSC-adj}
  \phi(\adjplus) = \frac{1}{\Delta} (\adjplus - \mu I).
\end{equation}
In this way, 
there is no need to redesign $S_U(\kappa)$- or $S_I(\kappa)$-specific filters for various tasks,
and thus existing mature spectral neural networks in collaborative filtering can be used to effectively
integrate graph-structured side information.

\textbf{Factor estimation.}
The shifting and scaling factors can be estimated automatically by leveraging approximation algorithms.
In the following, we illustrate how this procedure can be implemented efficiently using the classical power method.
Since $\lambda_{\max}(\tilde{A}) \equiv 1$, we are to estimate the minimum eigenvalue according to the Python-style pseudocode~\ref{alg-factor},
wherein \texttt{power\_iteration} is an eigenvalue algorithm that returns the largest (in absolute value) eigenvalue of the given matrix. Accordingly, we have:
$$
1 - \text{power\_iteration}(I - \tilde{A}_+) \approx 1 - (1 - \lambda_{\min}) = \lambda_{\min},
$$
and the `optimal' $\mu$ and $\Delta$.
Remark that an exact spectral correction may not necessarily be optimal. 
This is because the spectrum shift induced by the augmented graph is distributionally similar to a shifting and scaling transformation, but not completely equivalent. Indeed, the real transformation is considerably more intricate. Nonetheless, the estimated values of $\mu$ and $\Delta$ yield acceptable performance.
To achieve optimal results, it is necessary to retune them in the vicinity of the estimated values.

In a word, our method maintain any heuristically designed graph filters unchanged 
and make it match with the shifted oracle spectrum importance of augmented graph by shifting and scaling the augmented adjacency matrix.

\begin{algorithm}[t]
  \caption{
    Spectral GNNs equipped with SSC
  }
  \label{alg-convshift}
  \KwIn{
    spectrum shifting parameter $\mu$;
    spectrum scaling parameter $\Delta$;
    bipartite interaction graph $\adj$; 
    side information intensity $\kappa$;
    original graph filter $g(\cdot)$
  }

  Construct augmented graph $A_+$ by integrating $S_U(\kappa)$ or $S_I(\kappa)$ into the bipartite graph $A$;

  Obtain $\adjplus$ by performing symmetric sqrt normalization;

  Conduct SSC operations according to Eq.~\eqref{eq-SSC-adj};

  Model training based on the modified graph filter $g_+$.

  \KwOut{Final embeddings $H_+ = g_+ (\adjplus) E$}
\end{algorithm}

\subsection{Applications} \label{sec_adverse_examples}

\begin{figure}[t]
  \centering
  \includegraphics[width=0.47\textwidth]{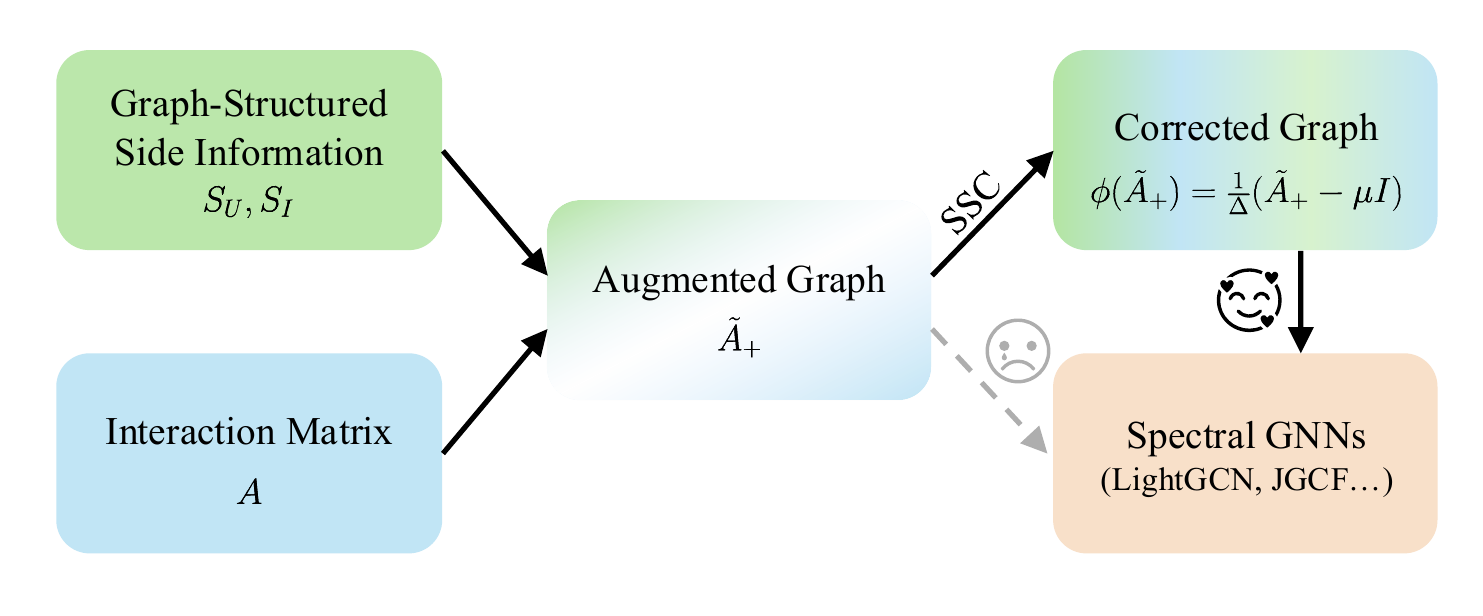}
  \caption{
    Pipeline of SSC.
  }
  \label{fig-pipeline}
\end{figure}

The universal procedures for boosting spectral GNNs are detailed in Algorithm~\ref{alg-convshift} and \figurename~\ref{fig-pipeline}.
Here, we introduce two specific examples.

\textbf{LightGCN} is a $L$-layer graph convolution network with $\alpha_l = 1 / (L + 1)$.
Applying SSC yields a variant as follows
\begin{equation}
  H = \sum_{l=0}^{L} \frac{1}{L+1} 
  \left(
    \frac{1}{\Delta}(\adjplus - \mu I)
  \right)^l E 
  \in \mathbb{R}^{(|\users| + |\items|) \times d}.
\end{equation}

\textbf{JGCF} further emphasizes the high frequencies based on the Jacobi Polynomial Bases~\cite{wang2022jacobiconv}:
\begin{align*}
  \mathbf{P}_l^{a, b}(\lambda) = (\theta_l \lambda + \theta_l') \mathbf{P}_{l-1}^{a, b}(\lambda) - \theta_{l}'' \mathbf{P}_{l-2}^{a, b}(\lambda), \: l \ge 2\\
  \mathbf{P}_0^{a, b}(\lambda) = 1, \quad
  \mathbf{P}_1^{a, b}(\lambda) = \frac{a - b}{2} + \frac{a + b + 2}{2} \lambda.
\end{align*}
Notably, $\theta_l, \theta_l', \theta_l''$ can be explicitly computed according to the shape hyperparameters $(a, b)$.
$\mathbf{P}_l^{a, b}(l), l=0, \ldots, L$ are orthogonal with weight function $(1 - \lambda)^a (1 + \lambda)^b$ on the interval [-1, 1].
Through spectrum shift correction, this orthogonality can potentially be preserved.
The resulting improved JGCF (the band-stop part) becomes
\begin{equation}
  H = \sum_{l=0}^L \frac{1}{L+1} \mathbf{P}_l^{a, b}
  \left(
    \frac{1}{\Delta}(\adjplus - \mu I)
  \right) E. 
\end{equation}

Given the embeddings $H$ derived from the modified LightGCN or JGCF, 
the relevance score between a user $u$ and an item $j$ can be calculated using the inner product:  
\begin{equation}
  y_{ui} = h_u^T h_i, \quad \forall u \in \mathcal{U}, i \in \mathcal{I}.
\end{equation}  
Items ranked at the top are then considered as promising candidates for recommendation.

\section{Experiments}

In this section, we evaluate SSC on 2 backbones and 5 datasets including multimodal recommendation and social recommendation.
The aim is to study the following research questions:
\begin{itemize}[leftmargin=*]
    \item RQ1: How does SSC perform when applied to spectral collaborative filtering methods?
    \item RQ2: What is the computational cost of SSC? Is it practical to apply it to spectral methods?
    \item RQ3: Does SSC effectively address spectrum shift and restore a full-spectrum alignment?
    \item RQ4: How well does SSC resist noise when integrating collaborative and side information, compared to state-of-the-art GNN-based methods?
\end{itemize}

\begin{table}[]
  \centering
  \caption{
    Dataset statistics.
    `\#U-I': the number of interactions.
    `\#U-U': the number of user-user edges;
    `\#I-I' is omitted as it is determined by the hyperparameter $\kappa$.
  }
\begin{tabular}{ccccc}
  \toprule
            & \#Users & \#Items & \#U-I     & \#U-U   \\
  \midrule
Ciao        & 6,804   & 17,660  & 139,112   & 106,433 \\
LastFM      & 1,875   & 4,613   & 77,465    & 25,182  \\
  \midrule
Baby        & 19,445  & 7,050   & 160,792   &   -     \\
Sports      & 35,598  & 18,357  & 296,337   &   -     \\
Electronics & 192,403 & 63,001  & 1,689,188 &   -     \\
  \bottomrule
\end{tabular}
\end{table}

\begin{table*}[]
  \centering
  \caption{
    Overall performance comparison for multimodal recommendation.
    The best results are marked in \textbf{bold}.
    Paired t-test is performed over 5 independent runs for evaluating $p$-value ($\;^*$ indicates statistical significance for $p \le 0.05$ , and $\;^{**}$ for $p \le 0.01$).
    Symbol $\blacktriangle$\% stands for the relative improvement against the backbones.
    `-' indicates that the method cannot be performed with an RTX 3090 GPU.
    The `Time' column reports the total training time.
  }
  \label{table-performance-multimodal}
  \scalebox{.8}{
    \setlength{\tabcolsep}{1.2mm}{
    \begin{tabular}{l|ccccc|ccccc|ccccc}
    \toprule
    & \multicolumn{5}{c|}{Baby} & \multicolumn{5}{c|}{Sports} & \multicolumn{5}{c}{Electronics} \\
     & R@10 & R@20 & N@10 & N@20 & Time & R@10 & R@20 & N@10 & N@20 & Time & R@10 & R@20 & N@10 & N@20 & Time \\
    \midrule
    MF-BPR & 0.0417 & 0.0638 & 0.0227 & 0.0284 & 04m15s & 0.0586 & 0.0885 & 0.0324 & 0.0401 & 09m59s & 0.0372 & 0.0557 & 0.0208 & 0.0256 & 45m17s \\
    \midrule
    \midrule
    ChebyNet & 0.0383 &	0.0609 & 0.0201	& 0.0259 & 05m10s & 0.0366 & 0.0564 & 0.0202 & 0.0253 & 09m37s & 0.0201 & 0.0309 & 0.0104 & 0.0132 & 133m19s\\
    JacobiConv & 0.0456	& 0.0716 & 0.0249 & 0.0316 & 05m14s & 0.0471 & 0.0730 & 0.0262 & 0.0329 & 09m06s & 0.0281 & 0.0418 & 0.0160 & 0.0195 & 152m54s \\
    \midrule
    \midrule
    MMGCN & 0.0353 & 0.0580 & 0.0189 & 0.0248 & 16m13s & 0.0331 & 0.0529 & 0.0179 & 0.0230 & 55m08s & 0.0217 & 0.0339 & 0.0117 & 0.0149 & 354m46s \\
    LATTICE & 0.0562 & 0.0875 & 0.0308 & 0.0388 & 10m57s & 0.0683 & 0.1024 & 0.0379 & 0.0466 & 82m45s & - & - & - & - & - \\
    BM3 & 0.0559 & 0.0866 & 0.0306 & 0.0385 & 06m35s & 0.0647 & 0.0980 & 0.0358 & 0.0444 & 16m28s & 0.0417 & 0.0631 & 0.0233 & 0.0289 & 149m07s \\
    FREEDOM & 0.0649 & 0.0991 & 0.0346 & 0.0434 & 07m22s & 0.0715 & 0.1088 & 0.0383 & 0.0479 & 16m53s & 0.0427 & 0.0647 & 0.0239 & 0.0295 & 129m32s\\
    MMSSL & 0.0595 & 0.0929 & 0.0350 & 0.0442 & 28m22s & 0.0667 & 0.1001 & 0.0390 & 0.0482 & 223m54s & - & - & - & - & - \\
    \midrule
    \midrule
    LightGCN & 0.0543 & 0.0851 & 0.0293 & 0.0372 & 04m42s & 0.0633 & 0.0956 & 0.0349 & 0.0432 & 11m35s & 0.0393 & 0.0579 & 0.0224 & 0.0272 &59m20s \\
    \quad \quad +SSC & 0.0666$^{**}$ & \textbf{0.1044}$^{**}$ & 0.0361$^{**}$ & \textbf{0.0458}$^{**}$ & 04m48s & 0.0762$^{**}$ & 0.1135$^{**}$ & 0.0417$^{**}$ & 0.0513$^{**}$ & 11m48s & \textbf{0.0452}$^{**}$ & \textbf{0.0677}$^{**}$ & \textbf{0.0254}$^{**}$ & \textbf{0.0312}$^{**}$ & 77m12s \\
    \multicolumn{1}{c|}{$\blacktriangle$\%} &22.61\%	&22.69\%	&23.33\%	&23.20\%	& -	&20.40\%	&18.77\%	&19.43\%	&18.78\%	& -	&15.01\%	&16.92\%	&13.31\%	&14.68\% & - \\
    \midrule
    JGCF & 0.0563 & 0.0874 & 0.0305 & 0.0385 & 05m06s & 0.0664 & 0.0999 & 0.0369 & 0.0456 & 09m21s & 0.0425 & 0.0618 & 0.0244 & 0.0294 & 140m10s\\ 
    \quad \quad +SSC & \textbf{0.0672}$^{**}$ & 0.1019$^{**}$ & \textbf{0.0364}$^{**}$ & 0.0453$^{**}$ & 05m16s & \textbf{0.0783}$^{**}$ & \textbf{0.1177}$^{**}$ & \textbf{0.0427}$^{**}$ & \textbf{0.0529}$^{**}$ & 09m36s & 0.0443$^{**}$ & 0.0648$^{**}$ & 0.0251$^{**}$ & 0.0304$^{**}$ & 166m34s \\
    \multicolumn{1}{c|}{$\blacktriangle$\%} & 19.34\% & 16.55\% & 19.31\% & 17.71\% & - & 17.92\% & 17.82\% & 15.73\% & 15.94\% & -  & 4.29\%	& 4.93\%	& 2.86\%	& 3.45\% & - \\
    \bottomrule
    \end{tabular}
    }
}
\end{table*}

\begin{figure*}[t]
  \centering

  \begin{minipage}[t]{0.7\textwidth}
  \centering
  \captionof{table}{
    Overall performance comparison for social recommendation.
  }
  \label{table-performance-social}
  \scalebox{.8}{
    \setlength{\tabcolsep}{1.2mm}{
    \begin{tabular}{l|ccccc|ccccc}
    \toprule
    & \multicolumn{5}{c|}{Ciao} & \multicolumn{5}{c}{LastFM} \\
     & R@10 & R@20 & N@10 & N@20 & Time & R@10 & R@20 & N@10 & N@20 & Time \\
    \midrule
    MF-BPR     & 0.0450 & 0.0737 & 0.0269 & 0.0353 & 03m44s & 0.2357 & 0.3220 & 0.2165 & 0.2513 & 02m15s \\
    \midrule
    \midrule
    ChebyNet   & 0.0466 & 0.0722 & 0.0300 & 0.0375 & 04m43s & 0.2052 & 0.2875 & 0.1848 & 0.2180 & 02m45s \\
    JacobiConv & 0.0468 & 0.0726 & 0.0303 & 0.0379 & 05m09s & 0.2177 & 0.3024 & 0.1966 & 0.2309 & 03m33s \\
    \midrule
    \midrule
    DiffNet++ & 0.0482       & 0.0774       & 0.0293       & 0.0379    &12m08s & 0.2434 & 0.3320 & 0.2245 & 0.2604 & 06m12s \\
    DSL       & 0.0532 & 0.0827       & 0.0323 & 0.0409    & 26m30s & 0.2495 & 0.3375 & 0.2302 & 0.2659 & 17m22s \\
    SHaRe     & 0.0519       & 0.0827 & 0.0319       & 0.0409  & 28m52s & 0.2442 & 0.3352 & 0.2228 & 0.2595 & 03m47s \\
    \midrule
    \midrule
    LightGCN   & 0.0473 & 0.0770 & 0.0295 & 0.0381 & 04m28s & 0.2527 & 0.3430 & 0.2326 & 0.2690 & 02m39s \\
    \quad \quad +SSC & \textbf{0.0547}$^{**}$ & \textbf{0.0861}$^{**}$ & \textbf{0.0340}$^{**}$ & \textbf{0.0431}$^{**}$ & 04m28s & \textbf{0.2588}$^{**}$ & 0.3502$^{**}$ & \textbf{0.2406}$^{**}$ & 0.2772$^{**}$ & 02m34s \\
    \multicolumn{1}{c|}{$\blacktriangle$\%} &15.71\%	&11.80\%	&15.30\%	& 13.06\%	& -	&2.43\%	&2.10\%	&3.42\%	&3.05\% & - \\
    \midrule
    JGCF       & 0.0498 & 0.0778 & 0.0316 & 0.0397 & 04m46s & 0.2555 & 0.3465 & 0.2369 & 0.2736 & 03m49s \\
    \quad \quad +SSC & 0.0543$^{**}$ & 0.0803$^{**}$ & 0.0336$^{**}$ & 0.0413$^{**}$ & 05m00s & 0.2583 & \textbf{0.3504}$^{*}$ & 0.2402$^{*}$ & \textbf{0.2772}$^{**}$ & 04m21s \\
    \multicolumn{1}{c|}{$\blacktriangle$\%} & 9.04\% & 3.20\% & 6.19\% & 3.91\% & - & 1.10\% & 1.12\% & 1.41\% & 1.32\% & - \\
    \bottomrule
    \end{tabular}
    }
    }
  \end{minipage}
  \hfill
  \begin{minipage}[t]{0.26\textwidth}
    \vspace{0.8em}
    \centering
    \includegraphics[width=\textwidth]{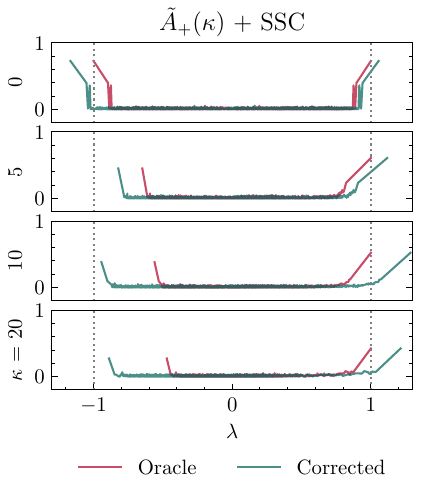}
    \captionof{figure}{Effect of SSC.}
    \label{fig-SSC-visualization}
  \end{minipage}

\end{figure*}

\subsection{Experiment Setup}

\textbf{Datasets.}
Experiments are conducted under two scenarios, 
encompassing social and multimodal recommendation.
For social datasets, we consider two real-world datasets, 
including Ciao\footnote[1]{https://www.cse.msu.edu/{\textasciitilde}tangjili/datasetcode/truststudy.htm}\cite{tang2012ciao1}\cite{tang2012ciao2} and LastFM\footnote[2]{http://www.lastfm.com}\footnote[3]{https://github.com/Sherry-XLL/Social-Datasets/tree/main/lastfm}\cite{cantador2011lastfm}.
% and Yelp\footnote[4]{https://github.com/Sherry-XLL/Social-Datasets/tree/main/yelp2}.
For data quantity consideration, we remove nodes with less than 3 interactions.
The user-user social networks are symmetrically derived from social relationships.
For multimodal datasets, three widely used datasets are considered, including Amazon Baby, Sports, and Electronics\footnote[4]{
https://github.com/enoche/MMRec
}, published by \cite{he2016ups}. 
We extract user-item interaction from the positive review histories (rating larger or equal to 4.0/5.0) and remove nodes with less than 5 interactions. 
Following \cite{zhou2023mmrec}, 
we use CNN-encoded thumbnails of each item as the visual feature 
and use Transformer (all-MiniLM-L6-v2) encoded descriptions as the textual feature.
For both scenarios, we split the preprocessed datasets into training sets, validation sets, and test sets with the ratio 8:1:1.

\textbf{Baselines.}
To validate the effectiveness of SSC,
we compare it with several research lines: 
% I) traditional collaborative filtering methods: MF-BPR~\cite{rendle2012bpr}, LightGCN~\cite{he2020lightgcn}, JGCF~\cite{guo2023jgcf};
% II) spectral GNNs with learnable polynomial filters: ChebyNet~\cite{defferrard2016chebynet}, JacobiConv~\cite{wang2022jacobiconv}; 
% III) state-of-the-art multimodal collaborative filtering methods: MMGCN~\cite{wei2019mmgcn}, LATTICE~\cite{zhang2021lattice}, BM3~\cite{Zhou2023bm3}, FREEDOM~\cite{zhou2023freedom}, MMSSL~\cite{wei2023mmssl}; 
% IV) state-of-the-art social collaborative filtering methods: DiffNet++~\cite{wu2020diffnet++}, DSL~\cite{wang2023dsl}, SHaRe~\cite{jiang2024share}. 
% Further details are provided in Appendix~\ref{appendix-baselines}.

\noindent I) traditional collaborative filtering methods:

\begin{itemize}[]
  \item MF-BPR~\cite{rendle2012bpr}: directly using a latent embedding to model each user and item and optimize by Bayesian pair-wise loss.
  \item LightGCN~\cite{he2020lightgcn}: A simplified recommendation framework based on Graph Neural Network and can be seen as a fixed polynomial spectral filter.
  \item JGCF~\cite{guo2023jgcf}: utilize Jacobi basis to filter low-frequency and high-frequency signals, making it fit for recommendation task.
\end{itemize}

\noindent II) spectral GNNs with learnable polynomial filters:

\begin{itemize}[]
  \item ChebyNet~\cite{defferrard2016chebynet}: A learnable graph filter on Chebyshev basis.
  \item JacobiConv~\cite{wang2022jacobiconv}: A learnable graph filter on Jacobi basis. The main difference with JGCF is the learnable parameter.
\end{itemize}

\noindent III) state-of-the-art multimodal collaborative filtering methods:

\begin{itemize}[]
  \item MMGCN~\cite{wei2019mmgcn}: aggregate modality-specific representation on each graph to incorporate multimodal features.
  \item LATTICE~\cite{zhang2021lattice}: construct learnable item-item similarity graph using multimodal features and capture modality information by item convolution on such graph.
  \item BM3~\cite{Zhou2023bm3}: generate contrastive views by dropout mechanism and inject modality features by alignment loss.
  \item FREEDOM~\cite{zhou2023freedom}: freeze the knn similarity graph introduced by LATTICE and propose a dynamic graph sampling strategy for the collaborative convolution.
  \item MMSSL~\cite{wei2023mmssl}: a self-supervised learning framework for multimodal recommendation aiming to enhance modality-aware user preference and cross-modal dependencies.
\end{itemize}

\noindent IV) state-of-the-art social collaborative filtering methods:

\begin{itemize}[]
  \item DiffNet++~\cite{wu2020diffnet++}: capture high-order collaborative interests and heterogeneous social influence on separate graphs, and then combine them attentively.
  \item DSL~\cite{wang2023dsl}: a typical self-supervised framework in social recommendation, which generates social view and collaborative view on each graph.
  \item SHaRe~\cite{jiang2024share}: rewrite user-user social graph dynamically to extract useful information. We implement the LightGCN + social + SHaRe version mentioned in their report.
\end{itemize}

\textbf{Evaluation metrics.}
Based on the top $N$ ranked candidate items determined by the relevance scores, 
we employ two popular metrics in recommender systems for performance evaluation: Recall@$N$ and NDCG@$N$. 
Specifically, Recall emphasizes the model's ability to retrieve positive items, 
while NDCG, which stands for Normalized Discounted Cumulative Gain, 
places greater emphasis on the ranking quality of positive items.
Following previous works~\cite{zhou2023freedom,wang2023dsl},
the values $N=10$ and $N=20$ are investigated in the subsequent analysis.

\textbf{Implementation details.}
For a fair comparison, we adopt the Adam optimizer~\cite{kingma2014adam} for training and fix embedding dimension to 64, batch size to 2048. 
We tune learning rates of all methods within \{1.e-4, 5.e-4, 1.e-3, 5.e-3\}. 
Other involved model/loss hyperparameters are tuned in accordance with the guidelines suggested by the baselines.
To enhance the spectral GNNs (\eg, LightGCN, JGCF) with the proposed SSC,
we retain the previously optimized hyperparameters and perform joint tuning of $\kappa$, $\mu$, and $\Delta$.
We search $\kappa$ from \{0, 5, 10, 15, 20\} for multimodal recommendation and \{0.0, 0.25, 0.5, 0.75, 1.0, 1.25, 1.5\} for social recommendation.
We tune $\mu$ within the range of 0.0 to 0.4 with a step length of 0.05, and $\Delta$ within the range of 0.4 to 1.0 with a step length of 0.05.

\subsection{Overall Performance Comparison (RQ1 \& RQ2)}
\label{overall-performance}
We present the overall performance of SSC on LightGCN and JGCF in Table~\ref{table-performance-multimodal} and~\ref{table-performance-social}, from which we have the following observations.
\begin{itemize}[leftmargin=*]
    \item 
    As a typical collaborative filtering method, LightGCN demonstrates significantly stronger performance compared to MF-BPR. For spectral collaborative filtering, JGCF surpasses LightGCN by effectively manipulating both low-frequency and high-frequency signals. As previously mentioned (see Table~\ref{table-learnable}), learnable graph filters (\eg, ChebyNet, JacobiConv) are not effective in recommender systems. When side information is incorporated, most baseline models outperform pure collaborative filtering methods, illustrating the utility of side information.
    \item 
    In multimodal recommendation, although most state-of-the-art baselines (\eg, FREEDOM, MMSSL) outperform traditional collaborative filtering models, they typically aggregate multimodal information on separate $\kappa$NN graphs rather than utilizing the augmented graph. Furthermore, these methods adopt equally weighted convolution layers due to conventional practices, neglecting the spectral influence on collaborative convolution. By leveraging SSC, we extend the applicability of spectral collaborative filtering methods to multimodal recommendation, allowing a shifted spectrum to align with the full-spectrum assumption. Not only does this approach extract valuable multimodal information, achieving improvements of 2.86\% to 23.33\%, but it also surpasses the performance of these baselines.
    \item 
    Similarly, none of the social recommendation baselines consider aligning the graph with the filter they employ. On the LastFM dataset, social recommendation methods (\eg, SHaRe, DSL) perform worse than LightGCN, indicating that the social relationships in this dataset are less informative. Although SHaRe incorporates an adaptive graph rewriting strategy, the results remain unsatisfactory. However, SSC enhances both LightGCN and JGCF by adjusting the spectrum to extract useful side information, resulting in improvements ranging from 1.10\% to 15.71\%.
    \item 
    GNN-based recommendation systems are often computationally expensive in industrial applications, and state-of-the-art models are even less practical due to their high resource consumption, as indicated by our measurements. In contrast, SSC preprocesses the convolution graph prior to training, achieving nearly the same computational efficiency as the backbone model. This is typically less complex than models specifically designed for multimodal or social recommendation. Consequently, SSC can be seamlessly integrated into other spectral methods to exploit various types of side information effectively.
\end{itemize}

\subsection{SSC Visualization (RQ3)}

% \begin{figure}[t]
%   \centering
%   \includegraphics[width=0.47\textwidth]{src/counter_shift.pdf}
%   \caption{
%     Effect of spectrum shift correction.
%   }
%   \label{fig-SSC-visualization}
% \end{figure}

In this section, we employ JGCF as the backbone to illustrate the impact of SSC on $\adjplus$. 
The shifting and scaling factors are jointly tuned under $\kappa = 0, 5, 10, 20$, with the optimal hyperparameter configuration determined based on NDCG@20 performance on the validation set.
Subsequently, the spectrum of $\adjplus(\kappa)$ after applying spectrum shift correction is visualized in Figure~\ref{fig-SSC-visualization} using green lines. 
While the original oracle spectrum shifts rightward, as indicated by the red lines, the adjusted spectra under various $\kappa$ values are more closely aligned within the range of $[-1, 1]$. 
This alignment more closely adheres to the full-spectrum assumption, enabling the backbone model’s robust filtering mechanism to be effectively leveraged for collaborative filtering  while incorporating side information.

\subsection{Robustness of SSC Against Noise (RQ4)}
\label{section-exp-noise}

\begin{figure}[t]
  \centering
  \includegraphics[width=0.47\textwidth]{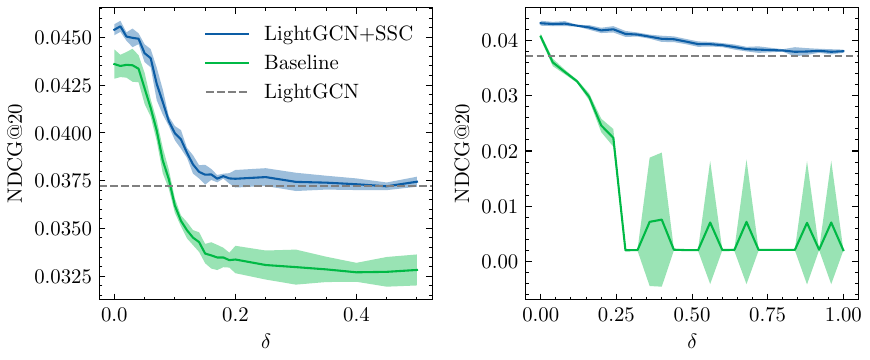}
  \caption{
    Model performance under progressively increasing levels of noise intensity $\delta$.
  }
  \vspace{-1em}
  \label{fig-noise}
\end{figure}

In this section, we examine the noise-resistance capability of SSC in comparison to typical GNN-based methods in both multimodal and social recommendation scenarios. 

For multimodal recommendation, we choose FREEDOM as a baseline for comparison, as it employs two separate branches for collaborative convolution and multimodal convolution. 
Gaussian noise $\epsilon \sim N(0, \delta)$ is added to the raw multimodal features, with $\delta$ increasing from 0 to 0.5. 
For social recommendation, SHaRe is selected as the baseline, which aggregates different types of information on distinct graphs. 
To simulate noise, edges in the graph are randomly altered to fake ones, with a fake edge ratio $\delta$ varying within $[0, 1]$.
We retain the same hyperparameters as optimized in Section~\ref{overall-performance} and introduce noise into the system. Each experiment is conducted five times to ensure reliability.

As illustrated in Figure~\ref{fig-noise}, LightGCN equipped with SSC demonstrates strong robustness, maintaining performance nearly identical to the backbone, even under conditions of high noise intensity. 
In contrast, FREEDOM experiences significant performance degradation when the noise intensity $\delta$ exceeds approximately 0.15. 
This issue likely arises from the parallel structure of collaborative and multimodal convolution. 
When the multimodal branch becomes inundated with noise, the inability to effectively balance these two information sources can severely compromise the final representation. 
However, our SSC employs a unified convolution framework, which prevents it from being overly influenced by excessive noise within the side information.
The robustness challenge is even more pronounced in SHaRe, as it incorporates self-supervised loss and leverages existing social relationships to guide social graph rewriting. 
Such processes may unintentionally generate additional noise, leading to a detrimental feedback loop. 

\subsection{Sensitivity Analysis}

\begin{table}[]
\centering
\caption{Performance under different values of $\mu$ ($\;^*$ indicates estimated value). }
\begin{tabular}{c|cccc}
\hline
$\mu$ & R@10 & R@20 & N@10 & N@20 \\
\midrule
0 & 0.0663 & 0.1033 & 0.0359 & 0.0454 \\
0.05 & \textbf{0.0669} & 0.1038 & \textbf{0.0361} & 0.0455 \\
0.1 & 0.0665 & 0.1043 & 0.0360 & 0.0457 \\
0.1203* & 0.0663 & \textbf{0.1045} & 0.0360 & \textbf{0.0458} \\
0.15 & 0.0652 & 0.1034 & 0.0353 & 0.0451 \\
0.2 & 0.0643 & 0.1019 & 0.0348 & 0.0444 \\
\bottomrule
\end{tabular}
\label{tab:baby-mu}
\end{table}

\begin{table}[]
\centering
\caption{Performance under different values of $\Delta$ ($\;^*$ indicates estimated value).}
\begin{tabular}{c|cccc}
\toprule
$\Delta$ & R@10 & R@20 & N@10 & N@20 \\
\midrule
1 & 0.0654 & 0.1032 & 0.0355 & 0.0452 \\
0.9 & 0.0657 & 0.1034 & 0.0358 & 0.0454 \\
0.8 & \textbf{0.0667} & 0.1038 & 0.0361 & \textbf{0.0456} \\
0.8797* & 0.0659 & \textbf{0.1042} & 0.0357 & 0.0455 \\
0.7 & 0.0666 & 0.1031 & \textbf{0.0362} & 0.0455 \\
0.6 & 0.0657 & 0.1006 & 0.0356 & 0.0446 \\
\bottomrule
\end{tabular}
\label{tab:baby-delta}
\end{table}

We present the sensitivity experiment results on the Baby dataset in Tables~\ref{tab:baby-mu} and~\ref{tab:baby-delta}, which also demonstrate the robustness of our rough estimation method introduced in 
Algorithm~\ref{alg-factor}.

\section{Related Work}

\subsection{Graph Collaborative Filtering}

GNNs, particularly spectral GNNs, 
exhibit superiority in collaborative filtering due to the bipartite graph structure inherent in the interaction data.
Since LightGCN~\cite{he2020lightgcn} simplifies NGCF~\cite{wang2019ngcf} by eliminating all trainable components except for the embeddings, 
more and more researchers have begun to recognize the significance of the spectrum~\cite{mao2021ultragcn,peng2022less}.
\citet{shen2021smoothing} unified a series of classic works from a spectral perspective, 
highlighting the low-frequency enhancement function of LightGCN.
Furthermore, the highest-frequency signals are mined in BSPM~\cite{choi2023blurring} and JGCF~\cite{guo2023jgcf} through graph diffusion and Jacobi polynomial bases, respectively.
However, these spectral GNNs are constrained to the bipartite interaction graph that adheres to the full-spectrum assumption. 
For the augmented graph examined in this paper, the underlying graph filters need to be modified to achieve optimal performance.

It is worth noting that spectral GNNs in the traditional graph community have evolved in a markedly different direction.
While early variants, 
such as APPNP~\cite{gasteiger2018appnp}, 
heuristically design graph filters, 
subsequent efforts focus on learning the coefficients automatically. 
Notable examples include ChebyNet~\cite{defferrard2016chebynet,he2022convolutional}, GPR-GNN~\cite{chien2021gprgnn}, JacobiConv~\cite{wang2022jacobiconv}, FavardGNN~\cite{guo2023favard}.
However, spectral GNNs with learnable coefficients encounter training difficulties and instability~\cite{he2020lightgcn,xu2023stablegcn}.
Consequently, the extension of graph collaborative filtering to the augmented graph requires further consideration.

\subsection{Multimodal Recommendation}
Modern media, enriched with diverse modalities (\eg, textual descriptions, visual thumbnails, etc.), 
offer highly valuable information for real-world recommendation tasks.
Significant efforts have been made to integrate multimodal features into the traditional collaborative filtering framework. 
For instance, VBPR~\cite{he2015vbpr} combines visual features extracted by a pretrained CNN for prediction. 
% VECF~\cite{chen2019vecf} utilizes VGG~\cite{simonyan2014vgg} and an attention mechanism to provide visual explanations for each recommendation. 
Moreover, MMSSL~\cite{wei2023mmssl} has emerged as one of the state-of-the-art methods through the application of self-supervised learning.

GNNs have also become indispensable components for integrating multimodal information.
MMGCN~\cite{wei2019mmgcn} directly utilize raw modal features by aggregating neighbor representations for each modality.
However, these raw features often contain noise that is unrelated to recommendation~\cite{xu2024stair}.
LATTICE~\cite{zhang2021lattice} constructs the item-item $\kappa$NN similarity graph and dynamically refines it over epochs, 
as this approach ensures that only the most relevant relationships are preserved.
The subsequent work, FREEDOM~\cite{zhou2023freedom}, 
simplifies the $\kappa$NN graph and introduces a standard baseline comprising two distinct convolutional branches.
However, their generalizability is constrained, 
requiring specific modifications to adapt to tasks with diverse characteristics, statistical properties, and signal-to-noise ratios.
Recently, GUME~\cite{lin2024gume} attempted to incorporate multimodal information by utilizing an augmented graph, 
similar to the one we defined in Section~\ref{sec_augmented_graph}. 
Notably, there is a growing trend toward leveraging side information in a structured and integrated manner. 
However, GUME empirically adopts an equally weighted convolution framework, 
which fails to align with the spectrum of the augmented graph. 
In this work, our analysis of the oracle spectrum importance reveals this discrepancy, 
and we propose a spectrum shift correction approach to address this issue.

\subsection{Social Recommendation}

In addition to multimodal features, social relationships serve as valuable side information frequently utilized in real-world recommendation scenarios. 
The primary task addressed in this work remains the prediction of user-item behaviors based on known user-item interaction histories and user-user relationships. 
Although non-graph-based methods such as SoRec~\cite{ma2008sorec}, Social Regularization~\cite{ma2011socreg}, and TrustMF~\cite{liu2013trustmf} exist, 
GNNs dominate this area due to the inherently graph-structured nature of social networks.
GraphRec~\cite{fan2019graphrec} first applies graph convolution to social recommendation.
Several subsequent approaches~\cite{liao2022sociallgn}~\cite{huang2021kgcn}~\cite{long2021smin} simplifies convolution framework and explore better convolutional ordering in terms of user-user and user-item graphs.
For instance, DiffNet++~\cite{wu2020diffnet++} models a social diffusion process on the heterogeneous social network
and an interest diffusion process on the bipartite interaction graph, respectively.
% SMIN~\cite{long2021smin} utilizes both social connections and multiple meta relationships.
Also, some methods intend to extract deeper information from graph topologies.
% MHCN~\cite{yu2021mhcn} classifies social relationships as several motifs and constructs hypergraph for convolution.
MHCN~\cite{yu2021mhcn} constructs hypergraph based on the motifs derived from the social and purchase relationships.
In view of the significant noise in social relationships, 
SHaRe~\cite{jiang2024share} refines the social graph by removing edges with low embedding similarity and adding edges with high similarity. 
As a representative self-supervised learning-based approach, 
DSL~\cite{wang2023dsl} constructs a social view and an interaction view based on the respective graphs. 
These methods empirically extend the LightGCN convolution to social recommendation without accounting for the spectrum shift phenomenon.
In contrast to the progressively more complex designs, 
we propose a universal solution by equipping existing spectral GNNs with a correction mechanism to match the shifted spectrum. 
This approach is significantly simpler and more effective.

\section{Conclusion and Future Work}

\textbf{Conclusion.}
We empirically and theoretically identify the spectrum shift phenomenon observed in the augmented graph, 
which hinders existing spectral graph collaborative filtering methods from fully exploiting the instrumental graph-structured side information.
The proposed Spectrum Shift Correction (SSC) operations enable spectral GNNs to adapt to the shifted spectrum.
SSC has been demonstrated to be effective, efficient, and robust, 
yielding significant performance gains across both multimodal and social recommendation scenarios.

\noindent\textbf{Limitation and future work.}
SSC is straightforward to implement 
and demonstrates remarkable robustness and generalizability across diverse tasks. 
However, it still necessitates the tuning of two additional hyperparameters.
Enabling graph filters to be learnable in collaborative filtering is reserved for future work as a more sophisticated and elegant solution.

\section{Acknowledgments}
This work was supported in part by National Natural Science Foundation of China (No. 92270119) and East China Normal University Academic Excellence Program for Future Leading Talents (Advanced Project, Grant No.2025ZY-012).

We thank all reviewers for their insightful feedback during the rebuttal period. 
We are especially grateful to Reviewer VXX2 for inspiring our methodology for hyperparameter estimation via the eigenvalue algorithm.

\bibliographystyle{ACM-Reference-Format}
%\bibliography{sample-base}
\bibliography{refs}

\appendix

\section{Oracle Spectra}

We plot the oracle spectrum importance across different values of $\kappa$ on the other 3 datasets. 
As illustrated in Figure~\ref{fig-oracle-spectrum-appendix}, all of them exhibit the spectrum shift phenomenon.

\begin{figure*}[t]
    \centering
    \hspace*{0.01\textwidth}
    \subfloat[]{\includegraphics[width=0.3\textwidth]{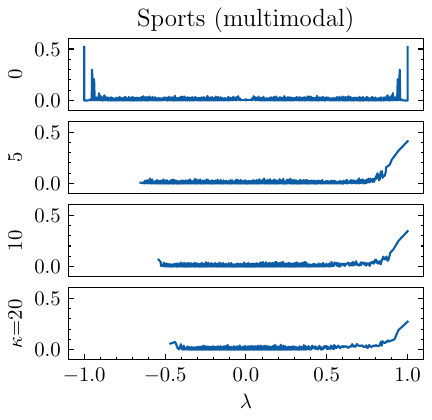}} \hfill
    \subfloat[]{\includegraphics[width=0.3\textwidth]{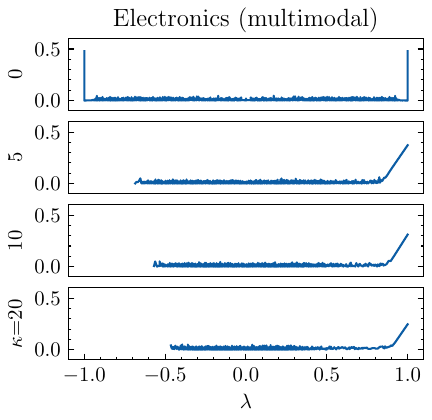}} \hfill
    \subfloat[]{\includegraphics[width=0.3\textwidth]{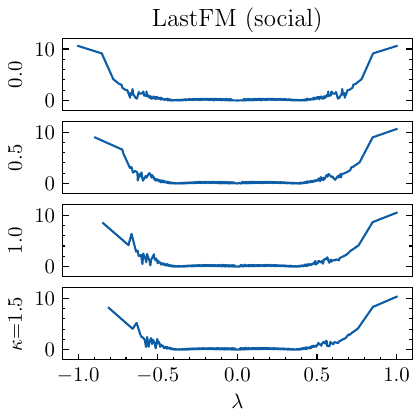}} 
    \hspace*{0.02\textwidth}
    
  \caption{
    Oracle spectrum importance $\mathcal{R}(U_{:, \lambda}; B)$ on other 3 datasets.
  }
  \label{fig-oracle-spectrum-appendix}
\end{figure*}

\section{Proofs}

\subsection{Full-Spectrum Assumption}
\label{appendix-proof-fsa}

For the sake of completeness, 
we provide a proof that the full-spectrum property holds for any bipartite graph:
\begin{equation}
   \adj = \left 
   [ \begin{matrix}
  0 & \tilde{R} \\
  \tilde{R}^T & 0 \\
  \end{matrix} \right ] \in \mathbb{R}^{(|\mathcal{U}| + |\mathcal{I}|) \times (|\mathcal{U}| + |\mathcal{I}|)},
\end{equation}
where $\tilde{R} = D_U^{-1/2} R D_I^{-1/2}$
and $D_U = \text{diag}(R \mathbf{1}_{|\mathcal{I}|}), D_I = \text{diag}(R^T \mathbf{1}_{|\mathcal{U}|})$.
Note that the proof of $-1 \leq \lambda \leq 1$ can be found in standard textbooks in terms of spectral graph~\cite{spielman2019spectral}. 
Here, we aim to demonstrate that the minimum and maximum eigenvalues are indeed achievable.
Define the following vectors
\begin{equation}
  v_+ = 
  \left[
    \begin{matrix}
      D_U^{1/2} \mathbf{1}_{|\mathcal{U}|} \\
      D_I^{1/2} \mathbf{1}_{|\mathcal{I}|} \\
    \end{matrix}
  \right],
  \quad
  v_- = 
  \left[
    \begin{matrix}
      D_U^{1/2} \mathbf{1}_{|\mathcal{U}|} \\
      -D_I^{1/2} \mathbf{1}_{|\mathcal{I}|} \\
    \end{matrix}
  \right].
\end{equation}
On the one hand,
$v_+$ is the unnormalized eigenvector corresponding to $\lambda_{\max} = 1$ since
\begin{align}
  \adj v_+
  & = 
  \left [ 
    \begin{matrix}
    0 & \tilde{R} \\
    \tilde{R}^T & 0 \\
    \end{matrix} 
  \right ]
  \left [
      \begin{matrix}
        D_U^{1/2} \mathbf{1}_{|\mathcal{U}|} \\
        D_I^{1/2} \mathbf{1}_{|\mathcal{I}|} \\
      \end{matrix}
  \right ] \\
  & = 
  \left [ 
    \begin{matrix}
    \tilde{R}D_I^{1/2} \mathbf{1}_{|\mathcal{I}|} \\
    \tilde{R}^T D_U^{1/2} \mathbf{1}_{|\mathcal{U}|} \\
    \end{matrix} 
  \right ]
  = 
  \left [ 
    \begin{matrix}
    D_U^{-1/2} R \mathbf{1}_{|\mathcal{I}|} \\
    D_I^{-1/2} R^T \mathbf{1}_{|\mathcal{U}|} \\
    \end{matrix} 
  \right ] \\
  &=
  \left [ 
    \begin{matrix}
    D_U^{-1/2} D_U \mathbf{1}_{|\mathcal{U}|} \\
    D_I^{-1/2} D_I \mathbf{1}_{|\mathcal{I}|} \\
    \end{matrix}
  \right ] 
  =
  \left [ 
    \begin{matrix}
    D_U^{1/2} \mathbf{1}_{|\mathcal{U}|} \\
    D_I^{1/2} \mathbf{1}_{|\mathcal{I}|} \\
    \end{matrix}
  \right ]  \\
  &= 1 \cdot v_+.
\end{align}
On the other hand, we have
\begin{align}
  \adj v_-
  & = 
  \left [ 
    \begin{matrix}
    0 & \tilde{R} \\
    \tilde{R}^T & 0 \\
    \end{matrix} 
  \right ]
  \left [
      \begin{matrix}
        D_U^{1/2} \mathbf{1}_{|\mathcal{U}|} \\
        -D_I^{1/2} \mathbf{1}_{|\mathcal{I}|} \\
      \end{matrix}
  \right ] \\
  & = 
  \left [ 
    \begin{matrix}
    -\tilde{R}D_I^{1/2} \mathbf{1}_{|\mathcal{I}|} \\
    \tilde{R}^T D_U^{1/2} \mathbf{1}_{|\mathcal{U}|} \\
    \end{matrix} 
  \right ]
  = 
  \left [ 
    \begin{matrix}
    -D_U^{-1/2} R \mathbf{1}_{|\mathcal{I}|} \\
    D_I^{-1/2} R^T \mathbf{1}_{|\mathcal{U}|} \\
    \end{matrix} 
  \right ] \\
  &=
  \left [ 
    \begin{matrix}
    -D_U^{-1/2} D_U \mathbf{1}_{|\mathcal{U}|} \\
    D_I^{-1/2} D_I \mathbf{1}_{|\mathcal{I}|} \\
    \end{matrix}
  \right ] 
  =
  \left [ 
    \begin{matrix}
    -D_U^{1/2} \mathbf{1}_{|\mathcal{U}|} \\
    D_I^{1/2} \mathbf{1}_{|\mathcal{I}|} \\
    \end{matrix}
  \right ] \\
  &= -1 \cdot v_-.
\end{align}

\subsection{Spectrum Shift}
\label{appendix-proof-ss}

\begin{proof}[The proof of Theorem~\ref{thm-spectrum-shift}:]
  Firstly, $\lambda_{\max} \equiv 1$ is a fixed point for any symmetric sqrt normalized adjacency matrix because
  \begin{align*}
    \adjplus D_{A_+}^{1/2} \mathbf{1}_{|\mathcal{U}| + |\mathcal{I}|}
    &= D_{A_+}^{-1/2} A_+ D_{A_+}^{-1/2} D_{A_+}^{1/2}  \mathbf{1}_{|\mathcal{U}| + |\mathcal{I}|} \\
    &= D_{A_+}^{-1/2} A_+  \mathbf{1}_{|\mathcal{U}| + |\mathcal{I}|} \\
    &= D_{A_+}^{-1/2} D_{A_+} \mathbf{1}_{|\mathcal{U}| + |\mathcal{I}|} \\
    &= 1 \cdot D_{A_+}^{1/2} \mathbf{1}_{|\mathcal{U}| + |\mathcal{I}|},
  \end{align*}
  where $D_{A_+} = \text{diag}(A_+ \mathbf{1}_{|\mathcal{U}| + |\mathcal{I}|})$.
  On the other hand,
  as $\kappa \rightarrow +\infty$, the augmented adjacency matrix reduces to
  \begin{equation}
    \adjplus \left \{
      \begin{array}{ll}
        = \left [
        \begin{matrix}
          \tilde{S} & 0 \\
          0 & 0
        \end{matrix}
        \right ],
         & \text{if } \kappa\text{-rescaling}, \\
        \approx \left [
        \begin{matrix}
          \frac{1}{|\mathcal{U}|} \mathbf{1}_{|\mathcal{U}|} \mathbf{1}_{|\mathcal{U}|}^T & 0 \\
          0 & 0
        \end{matrix}
        \right ],
         & \text{if } \kappa\text{-nearest-neighbors}.
      \end{array}
    \right .
  \end{equation}
  Note that the $\kappa$-nearest-neighbors case approximately holds true because for $\kappa$NN graphs, 
  $\kappa$ cannot be greater than the number of nodes (\ie, $\kappa \le |\mathcal{U}|$).
  However, in real-world recommendation scenarios, $\kappa$ could be significantly large due to the involvement of millions of nodes.
  Consequently, under $\kappa$-scaling, 
  the eigenvalues consist of the eigenvalues of $\tilde{S}$ along with an additional zero.
  As for the $\kappa$-nearest-neighbors case, the submatrix $\frac{1}{|\mathcal{U}|} \mathbf{1}_{|\mathcal{U}|} \mathbf{1}_{|\mathcal{U}|}^T$ is 
  positive semidefinite with two distinct eigenvalues: 0 and 1.
\end{proof}

\section{Experiment Details}

\subsection{Graph-structured Side Information Preprocessing}
\label{appendix-preprocessing}

This study involves two distinct graph construction approaches,
namely $\kappa$-rescaling for $S_U(\kappa)$ and $\kappa$-nearest-neighbors for $S_I(\kappa)$.
\begin{itemize}[leftmargin=*]
  \item 
  $\kappa$-rescaling multiplies the original social network $S_U$ by a coefficient $\kappa \in \mathbb{R}$, allowing the significance to be effectively controlled.
  Here, $S_U$ includes an edge between two users if certain social relationship of interest, such as friendships, exists between them.
  \item
  $\kappa$-nearest-neighbors is widely used in multimodal collaborative filtering~\cite{zhang2021lattice,zhou2023freedom} for constructing $\kappa$NN similarity graphs. Typically,
  \begin{equation*}
    S_I(\kappa) := \frac{1}{|\mathcal{M}|} \sum_{m \in \mathcal{M}} w_m \cdot S_m(\kappa), \quad \kappa=0, 1, 2, \ldots,
  \end{equation*}
  where $S_m(\kappa), m \in \mathcal{M}$ is the $\kappa$NN similarity graph corresponding the modality $m$,
  and $w_m \in [0, 1]$ denotes the respective weight.
  Here, the value of $\kappa$ signifies that $\kappa$ nearest nodes are selected as neighbors for each item.
  Given the raw encoded multimodal features, item similarity is typically calculated using the inner product between these features.
\end{itemize}

\subsection{Hyperparameter Settings}

\begin{table}[]
\centering
\caption{Hyperparameters used for different backbones on each dataset.}
\begin{tabular}{l|ccc|ccc}
\toprule
 & \multicolumn{3}{c|}{LightGCN} & \multicolumn{3}{c}{JGCF} \\
 & $\kappa$ & $\mu$ & $\Delta$ & $\kappa$ & $\mu$ & $\Delta$ \\
\midrule
Baby       & 10       & 0.05       & 0.75       & 5       & 0.05       &  0.85      \\
Sports       & 5       &  0.1      &  0.7      &  5      &  0.05      &  0.85      \\
Electronics      &  10      & 0.0       &  0.55      &  5      &  0.0      &  0.75      \\
\midrule
Ciao       & 0.75       & 0.15       &  0.4      &  0.75      &  0.0      &  0.55      \\
LastFM       & 1.0       & 0.0       & 0.5       &  0.75      &  0.0      &  1.0      \\
\bottomrule
\end{tabular}
\label{tab:hyperparams}
\end{table}

Table~\ref{tab:hyperparams} summarizes the tuned hyperparameters across various dataset and model combinations.

\section{Additional Experiments}

\subsection{Other Baselines and Backbones}

\begin{table}[]
  \centering
  \caption{
    Additional performance comparison for more baselines and more backbones.
  }
  \label{table-performance-additional}
  \scalebox{.8}{
    \setlength{\tabcolsep}{1.2mm}{
    \begin{tabular}{l|cccc|cccc}
    \toprule
    & \multicolumn{4}{c|}{Baby} & \multicolumn{4}{c}{Ciao} \\
     & R@10 & R@20 & N@10 & N@20 & R@10 & R@20 & N@10 & N@20 \\
    \midrule
    GF-CF     & 0.0350 & 0.0629 & 0.0130 & 0.0202 & 0.0014 & 0.0028 & 0.0007 & 0.0012 \\
    SVD-GCN   & 0.0513 & 0.0802 & 0.0278 & 0.0352 & 0.0471 & 0.0740 & 0.0299 & 0.0378 \\
    \midrule
    SMORE     & 0.0625 & 0.0946 & 0.0344 & 0.0427 & - & - & - & - \\
    STAIR     & 0.0674 & 0.1042 & 0.0359 & 0.0453 & - & - & - & - \\
    \midrule
    \midrule
    LightGCN  & 0.0543 & 0.0851 & 0.0293 & 0.0372 & 0.0473 & 0.0770 & 0.0295 & 0.0381 \\
    \quad \quad +SSC & 0.0666 & 0.1044 & 0.0361 & 0.0458 & 0.0547 & 0.0861 & 0.0340 & 0.0431 \\
    \multicolumn{1}{c|}{$\blacktriangle$\%} &22.61\%	&22.69\%	&23.33\%	& 23.20\% &15.71\%	&11.80\%	&15.30\%	& 13.06\%\\
    \midrule
    AFDGCF    & 0.0552 & 0.0861 & 0.0297 & 0.0377 & 0.0483 & 0.0779 & 0.0301 & 0.0387 \\
    \quad \quad +SSC & 0.0676 & 0.1037 & 0.0366 & 0.0459 & 0.0536 & 0.0851 & 0.0329 & 0.0421 \\
    \multicolumn{1}{c|}{$\blacktriangle$\%} &22.61\%	&20.48\%	&23.26\%	& 21.79\% &10.85\%	&9.15\%	&9.41\%	& 8.87\%\\
    \midrule
    SGL       & 0.0550 & 0.0861 & 0.0299 & 0.0379 & 0.0475 & 0.0759 & 0.0296 & 0.0380 \\
    \quad \quad +SSC & 0.0673 & 0.1024 & 0.0363 & 0.0454 & 0.0540 & 0.0861 & 0.0336 & 0.0429 \\
    \multicolumn{1}{c|}{$\blacktriangle$\%} &22.43\%	&18.86\%	&21.66\%	& 19.78\% &13.84\%	&13.36\%	&13.52\%	& 13.06\%\\
    \midrule
    SimGCL    & 0.0513 & 0.0802 & 0.0278 & 0.0352 & 0.0471 & 0.0740 & 0.0299 & 0.0378 \\
    \quad \quad +SSC & 0.0676 & 0.1048 & 0.0366 & 0.0462 & 0.0533 & 0.0852 & 0.0334 & 0.0427 \\
    \multicolumn{1}{c|}{$\blacktriangle$\%} &22.12\%	&21.21\%	&23.18\%	& 22.29\% &8.02\%	 &8.52\%	&9.27\%	& 9.24\%\\
    \bottomrule
    \end{tabular}
    }
}
\end{table}

To comprehensively demonstrate the superiority of SSC over competitive baselines, we compare it with several representative methods, including direct eigenvalue decomposition (GF-CF~\cite{shen2021smoothing}, SVD-GCN~\cite{peng2022svdgcn}) and recent state-of-the-art multimodal recommendation models (SMORE~\cite{ong2025smore}, STAIR~\cite{xu2024stair}). 
To further showcase its generalizability, we apply SSC to several alternative backbone architectures (AFDGCF~\cite{wu2024afdgcf}, SGL~\cite{wu2021sgl}, SimGCL~\cite{yu2022simgcl}), validating its capability to provide orthogonal improvements across different model designs. 
The results are shown in Table~\ref{table-performance-additional}.

\subsection{Larger datasets}

\begin{table}[]
\centering
\caption{Dataset statistics.}
\label{table:additional-datasets}
\begin{tabular}{ccccc}
  \toprule
            & \#Users & \#Items & \#Interactions  \\
  \midrule
    Books        & 603,668   & 367,982  & 8,898,041  \\
    Yelp         & 186,328   & 94,570   & 2,127,869  \\
  \bottomrule
\end{tabular}
\end{table}

\begin{table}[]
\centering
\caption{Performance comparison on Books.}
\begin{tabular}{l|cccc}
\toprule
 & R@10 & R@20 & N@10 & N@20 \\
\midrule
FREEDOM    & OOM    & OOM    & OOM    & OOM    \\
SMORE      & OOM    & OOM    & OOM    & OOM    \\
\midrule
LightGCN   & 0.0596 & 0.0878 & 0.0348 & 0.0423 \\
\quad \quad +SSC  & 0.0632 & 0.0932 & 0.0369 & 0.0449 \\
\multicolumn{1}{c|}{$\blacktriangle$\%}  & 6.04\% & 6.15\% & 6.03\% & 6.15\% \\
\bottomrule
\end{tabular}
\label{tab:additional-exp-books}
\end{table}

\begin{table}[]
\centering
\caption{Performance comparison on Yelp.}
\begin{tabular}{l|cccc}
\toprule
 & R@10 & R@20 & N@10 & N@20 \\
\midrule
DSL        & 0.0674 & 0.1055 & 0.0367 & 0.0468 \\
SHaRe      & OOM    & OOM    & OOM    & OOM    \\
\midrule
LightGCN   & 0.0645 & 0.1010 & 0.0352 & 0.0448 \\
\quad \quad +SSC  & 0.0689 & 0.1079 & 0.0379 & 0.0483 \\
\multicolumn{1}{c|}{$\blacktriangle$\%}  & 6.87\% & 6.87\% & 7.69\% & 7.83\% \\
\bottomrule
\end{tabular}
\label{tab:additional-exp-yelp}
\end{table}

We evaluate SSC on two larger datasets: Amazon Books\footnote[1]{https://github.com/enoche/MMRec}~\cite{he2016ups} for multimodal recommendation and Yelp\footnote[2]{https://business.yelp.com/data/resources/open-dataset} for social recommendation.
Dataset statistics are presented in Table~\ref{table:additional-datasets}, and the corresponding results are reported in Tables~\ref{tab:additional-exp-books} and~\ref{tab:additional-exp-yelp}.

\end{document}